%% file: CCCG.tex
\documentclass{cccg22}
\usepackage{graphicx,amssymb,amsmath}
\usepackage{algorithm,algorithmic}

\usepackage{subcaption}
\usepackage{cite}
\usepackage{cleveref}

\newcommand{\eps}{\varepsilon}                          
\newcommand{\epsapprox}{$(1+\eps)$-approximation}
\newcommand{\dc}{d_C} 
\renewcommand{\emph}[1]{\textit{\textbf{#1}}}
\usepackage{xcolor}
\iftrue
\newcommand\comment[1]{\textcolor{red}{#1}}
\newcommand\xcomm[2]{\comment{#1 says: #2}}
\newcommand\gill[1]{\xcomm{GB}{#1}}
\else
\newcommand\comment[1]{}
\newcommand\xcomm[2]{}
\newcommand\gill[1]{}
\fi
\DeclareRobustCommand*{\ora}{\overrightarrow}
\newcommand{\MWA}{{\rm MWA}}
\newcommand{\MinBall}{{\rm MinBall}}
\newcommand{\MaxBall}{{\rm MaxBall}}
\newcommand{\SP}{\kern+1pt} 

\newtheorem{remark}{Remark}

\title{Diamonds are Forever in the Blockchain:\\ Geometric Polyhedral Point-Set Pattern Matching}
\author{
   \quad\quad\quad
   Gill Barequet\thanks{ 
   Dept.\ of Computer Science,
   Technion---Israel Inst.\ of Technology,
 {\tt barequet@cs.technion.ac.il}}
   \quad\quad\quad
\and
   \quad\quad\quad
   Shion Fukuzawa\thanks{ 
 Dept.\ of Computer Science, University of California, Irvine,
 {\tt \{fukuzaws,goodrich,mosegued,eozel\}@uci.edu}}
   \quad\quad\quad
\and
   \quad\quad\quad
   Michael T. Goodrich$^\dagger$
   \quad\quad\quad
\and
   David M. Mount\thanks{ 
   Dept.\ of Computer Science,
 University of Maryland, College Park,
 {\tt mount@umd.edu}}
   \quad\quad\quad
\and
   \quad
   Martha C. Osegueda$^\dagger$
   \quad\quad\quad
\and
   \quad\quad\quad
   Evrim Ozel$^\dagger$
   \quad\quad\quad
}

\begin{document}
\thispagestyle{empty}
\maketitle

\begin{abstract}
Motivated by blockchain technology for supply-chain tracing of
ethically sourced diamonds, we 
study
geometric polyhedral point-set pattern matching 
as minimum-width polyhedral annulus problems 
under translations and rotations. 
We provide two \hbox{\epsapprox} schemes under translations 
with $O(\eps^{-d} n)$-time for~$d$ dimensions and $O(n\log \eps^{-1} + \eps^{-2})$-time for two dimensions, and we give
an $O(f^{d-1}\eps^{1-2d}n)$-time algorithm when 
also allowing for rotations, parameterized on $f$, which we define as the slimness of the point set.
\end{abstract}

\input{./intro}
\input{./preliminaries}

\input{./placement}

\input{./rotations}


\small
\bibliographystyle{abbrv}

\bibliography{refs}

\end{document}

%% file: intro.tex
\section{Introduction}

A notable recent computational geometry application 
is for tracking supply chains for natural diamonds, 
for which the industry and customers
are strongly motivated to prefer ethically-sourced provenance (\textit{e.g.}, to avoid
so-called ``blood diamonds'').
For example, the \emph{Tracr} system employs a blockchain for
tracing the supply chain for a diamond from its being mined as a rough
diamond to a customer purchasing a polished diamond~\cite{tracr}.
(See Figure~\ref{fig:tracr}.)

\begin{figure}[hbt]
\vspace*{-2pt}
\centering
\includegraphics[width = 3.3in, trim = 0.7in 1.5in 0.2in 1.3in, clip]{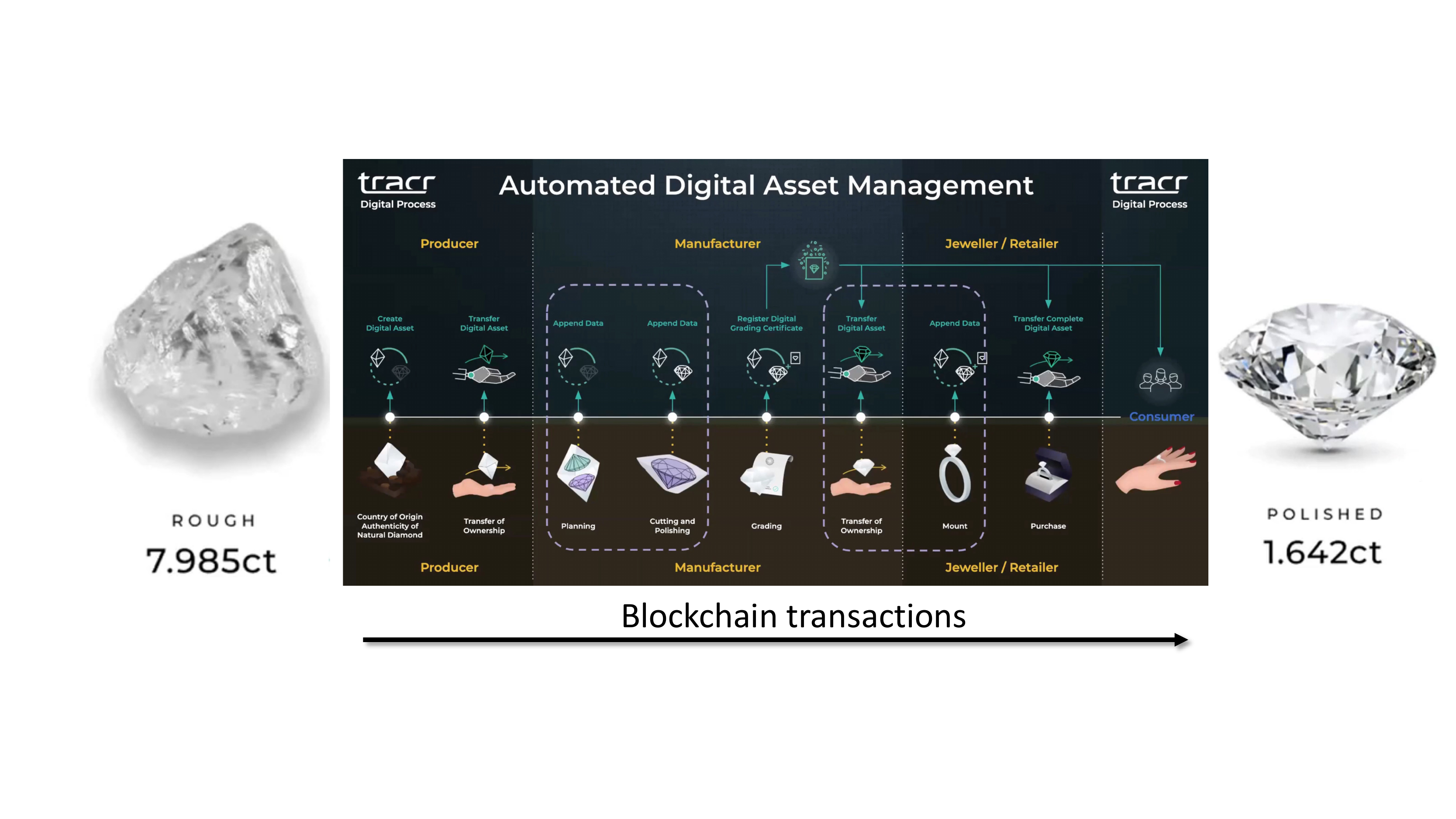}
\vspace*{-4pt}
\caption{\label{fig:tracr} Blockchain transactions in a diamond
supply chain, providing provenance, traceability, and authenticity
of an ethically-sourced diamond.
}
\vspace*{-4pt}
\end{figure}

Essential steps in the Tracr blockchain supply-chain
process require methods to match point sets against geometric shapes,
\textit{e.g.}, to guarantee that a diamond has not been replaced with one of
questionable provenance~\cite{tracr}. Currently, the Tracr system uses standard machine-learning techniques
to perform the shape matching steps. We 
believe, however, that better accuracy can be 
achieved by using computational geometry approaches.
In particular, motivated by the Tracr application,
we are interested in this paper in efficient methods for matching point sets
against geometric shapes, such as polyhedra.
Formalizing this problem, we study the problem of 
finding the best translation and/or rotation of
the boundary of a convex polytope, $P$ (\textit{e.g.}, defining a polished diamond shape),
to match a set of~$n$ points in a $d$-dimensional
($d \geq 3$) space, where the point set is a ``good'' sample of the
boundary of a polytope that is purported to be~$P$.
Since there may be small inaccuracies in the sampling process, our
aim is to compute a \emph{minimum width} polyhedral annulus determined by $P$ that
contains the sampled points.
In the interest of optimizing running time, rather than seeking an exact solution,
we seek an approximate solution that deviates
from the real solution by a predefined quantity~$\eps>0$.

\paragraph{Related Work.}
We are not familiar with any previous work on the problems we study in this
paper.
Nevertheless, there is considerable prior work on
the general area of matching a geometric shape to a set of points, especially
in the plane.
For example, Barequet, Bose, Dickerson, and 
Goodrich~\cite{DBLP:journals/jda/BarequetBDG05} give solutions 
to several constrained polygon annulus placement problems for
offset and scaled polygons
including an algorithm for finding the translation for the 
minimum offset of an $m$-vertex polygon that contains a set of $n$
points in $O(n\log^2 n + m)$ time.
Barequet, Dickerson, and Scharf~\cite{barequet2008covering}
study the problem of covering a maximum
number of $n$ points with an $m$-vertex polygon (not just its boundary) under
translations, rotations, and/or scaling, giving, \textit{e.g.}, an algorithm running
in time $O(n^3m^4\log (nm))$ for the general problem.
There has also been work on finding a minimum-width annulus
for rectangles and squares, \textit{e.g.},
see~\cite{GLUCHSHENKO2009168,BAE20182,BAE2021101697,MUKHERJEE201374}.

Chan~\cite{chan2000approximating} presents
a $(1+\eps)$-approximation method that finds a minimum-width spherical
annulus of $n$ points in $d$~dimensions
in $O(n\log (1/\eps)+\eps^{O(1)})$ time, and 
Agarwal, Har-Peled, and Varadarajan~\cite{a-core-04}
improve this to $O(n+1/\eps^{O(d^2)})$ time via 
coresets~\cite{agarwal2005geometric,phillips2017coresets,yu2008practical,agarwal2008robust}. 
A line of work has considered computing the spherical annulus under stronger assumptions on the points samples. Most notably Devillers and Ramos~\cite{devillers2002round} combine various definitions for ``minimum quality assumptions'' by Melhorn, Shermer and Yap~\cite{mehlhorn1997} and Bose and Morin~\cite{bose1998} and show that under this assumption the spherical annulus can be computed in linear time for $d=2$ and present empirical evidence for higher dimensions. 
Arya, da~Fonseca, and Mount~\cite{arya2018approximate} show how to find
an $\eps$-approximation of the width of $n$ points in 
$O(n\log (1/\eps) + 1/\eps^{(d-1)/2+\alpha})$ time, for a 
constant $\alpha>0$.
Bae~\cite{BAE2019398} shows how to find a min-width $d$-dimensional 
hypercubic shell in $O(n^{\lfloor d/2\rfloor}\log^{d-1} n)$ expected time.

\paragraph{Our Results.}
Given a set of $n$ points in ${\bf R}^d$, we provide an $O(\eps^{-d} n)$-time
$(1+\eps)$-approximate polytope-matching 
algorithm under translations, for $d \geq 3$, and 
$O(n\log \eps^{-1} + \eps^{-2})$ time for $d=2$, 
and we provide an $O(f^{d-1}\eps^{1-2d}n)$-time algorithm when also allowing
for rotations, where the complexity of the polytope is constant and for rotations is parameterized by $f$, which we define as the \emph{slimness} of the point set.

The paper is organized as follows.
In \Cref{sec:preliminaries}, we set the ground for this work by providing some
necessary definitions.
In \Cref{sec:search-space}, we approximate the \MWA\ under only translations. In this section, we provide a constant factor approximation scheme, a \epsapprox\ scheme and describe how to improve the running time in two dimensions. In \Cref{sec:rotations}, we consider the \MWA\ under rotations.

%% file: preliminaries.tex
\section{Preliminaries}
\label{sec:preliminaries}
Following previous
convention~\cite{amenta2007size,attali2003complexity,attali2004linear,erickson2001nice,amenta1999surface},
we say that
a point set~$S$ is a
\emph{$\delta$-uniform sample} of a surface $\Sigma \subset \mathbb{R}^d$
if for every point $p \in \Sigma$, there exists a point $q \in S$ such that
$d(p, q) \leq \delta$.
Let $C \subset \mathbb{R}^d$ be a closed, convex polyhedron containing the origin in its interior. Given $C$, and $x \in \mathbb{R}^d$, define $x + C = \{x + y : y \in C\}$ (the translation of~$C$ by~$x$), and for $r \in \mathbb{R}$, define $r C = \{r y : y \in C\}$. A \emph{placement} of $C$ is a pair $(x,r)$, where $x \in \mathbb{R}^d$ and $r \in \mathbb{R}^{\geq 0}$, representing the translated and scaled copy $x + r C$. We refer to $x$ and $r$ as the \emph{center} and \emph{radius} of the placement, respectively. Two placements are \emph{concentric} if they share the same center.

Let $C$ be any closed convex body in  $\mathbb{R}^d$ containing the origin in its interior. The convex distance function induced by $C$ is the function $d_C: \mathbb{R}^d \times \mathbb{R}^d \rightarrow \mathbb{R}^{\geq 0}$, where 
\[d_C(p,q) = \min \{r : r \geq 0 ~\mathrm{and}~ q \in p + r C\}\]

Thus, the convex distance between $p$ and $q$ is determined by the minimum radius placement of $C$ centered at $p$ that contains $q$ (see Figure~\ref{fig:prelims}). When $C$ is centrally symmetric, this defines a metric, but for general $C$, the function $d_C$ may not be symmetric. We call the original shape $C$ the \emph{unit ball} $U_C$ under the distance function $d_C$.
Note that $d_C(a,c)= d_C(a,b)+d_C(b,c)$ when $a$, $b$ and $c$ are collinear and appear in that order.

Define an \emph{annulus} for $C$ to be the set-theoretic difference of two concentric placements $(p + R C) \setminus (p + r C)$, for $0 \leq r \leq R$. The \emph{width} of the annulus is $R - r$.
Given a $\delta$-uniform sample of points, $S$, 
there are three placements of $C$ we are interested in:

$\bullet$ \textbf{Minimum enclosing ball (MinBall)}:  A placement of $C$ of the smallest radius that contains all of the points in $S$. 

$\bullet$ \textbf{Maximum enclosed ball (MaxBall)}: A placement of $C$ of the largest radius, centered within the convex hull of $S$, that contains no points in~$S$.

$\bullet$ \textbf{Minimum width annulus (MWA)}: A placement of an annulus for $C$ of minimum width, that contains all of the points in $S$. 

Note that, following the definition of the MaxBall,
we require that
the center of the MWA must also lie within the convex hull of $S$.
For each of the above placements, we also refer to parameterized versions, for example MinBall($p$), MaxBall($p$), or MWA($p$). 
These respectively refer to the minimum enclosing ball, maximum enclosed ball, or minimum width annulus {centered at the point $p$}.

Further, we use $|\MinBall(p)|$ and $|\MaxBall(p)|$ to 
denote the radius of MinBall$(p)$ and MaxBall$(p)$, respectively, 
and we use $|\MWA(p)|$ to denote the width of MWA$(p)$.

The ratio, $F$, of the MinBall over the MaxBall of 
$S\subset \mathbb{R}^d$ under distance function $d_C$ defines 
the \emph{fatness} of $S$ under $d_C$, such that $F:=|\MinBall|/|\MaxBall|$. Also, we define 
the \emph{concentric fatness} as the ratio of the MinBall and 
MaxBall centered at the MWA, such that $F_c:=|\MinBall(c_{opt})|/|\MaxBall(c_{opt})|$ where $c_{opt}$ is the center of the \MWA. Conversely, 
we define the \emph{slimness} 
to be $f^{-1}=1-F_c^{-1}$, which
corresponds to the ratio of the $\MinBall(c_{opt})$ over the MWA, \emph{i.e.}, $f:=|\MinBall(c_{opt})|/|\MWA|$. 
\begin{figure}
    \centering
    \includegraphics[width=0.95\linewidth]{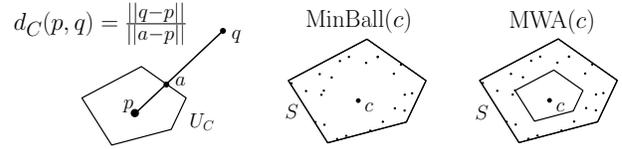}
    \caption{\textbf{Left:} a visual representation of a polyhedral distance function and the distance between two points. \textbf{Center:} The MinBall under $d_C$ containing all points in $S$, centered at $c$. \textbf{Right:} The MWA of $S$ with all points within MinBall$(c)\setminus $MaxBall(c).}
    \label{fig:prelims}
\end{figure}

\begin{remark} 
\label{rm:deltauniformity}
In order for a $\delta$-uniform sample to represent the surface, $\Sigma$, with sufficient accuracy for a meaningful $\MWA$, we assume that the sample must contain at least one point between corresponding facets of the $\MWA$. Where corresponding facets refer to facets of the $\MinBall$ and $\MaxBall$ representing the same facet of $U_C$. Therefore, 
in the remainder of the paper, we assume we have a $\delta$-uniform sample and 
that $\delta$ is small enough to guarantee this condition for even the smallest facets. 

\end{remark}

In practice, it would be easy to determine a small enough $\delta$ before sampling $\Sigma$, since only sufficiently slim surfaces would benefit from finding the $\MWA$, and very fat surfaces would yield increasingly noisy $\MaxBall$. One easy approach would be setting $\delta$ to the smallest facet of the $\MinBall$ and scaling down by an arbitrary constant larger than the maximum expected fatness, such as 100. This example imposes a very generous bound on fatness since it would allow the inner shell to be 1\% of the size of the outer shell, practically a single digit constant would often suffice.

Also, note that, for a given center point $c$, $\MWA(c)$ is uniquely defined as the annulus centered at $c$ with inner radius $\min_{p\in S} d_C(c,p)$ and outer radius $\max_{p\in S} d_C(c,p)$.
Further, let us assume
that the reference polytope defining our polyhedral distance function
has $m$ facets, where $m$ is a fixed constant,
since the sample size is expected to be 
much larger than $m$. Thus,
$\dc$ can be calculated in $O(m)$ time; hence, $\MWA(c)$ can be 
found in $O(mn)$ time, which is $O(n)$ under our assumption.

%% file: placement.tex
\section{Approximating the Minimum Width Annulus}\label{sec:search-space}

\input{./approx}

\input{./faster}

%% file: approx.tex
Let us first describe how to find a constant factor approximation
of $\MWA$ under translations.
Note that,
by assumption, 
the center $c$ of our approximation lies 
within the convex hull of $S$. 
Let us denote the center, outer radius, inner radius, and width of the optimal MWA as $c_{opt}$, $R_{opt}$, $r_{opt}$, and $w_{opt}$.

We begin with Lemma~\ref{lm:searchspace}, where we prove $c_{opt}$ is within a certain distance from the center of the \MinBall\, $c$, providing a search region for $c_{opt}$. In Lemma~\ref{lm:dist}, we bound the width achieved by a center-point that is sufficiently close to $c_{opt}$. We then use this in Lemma~\ref{lm:capprox} to prove that $|\MWA (c)|$ achieves a constant factor approximation.

\begin{lemma}
\label{lm:searchspace}
The center of the $\MWA$, $c_{opt}$, is within distance $w_{opt}$ of the 
center of the $\MinBall$, $c$. That is, $\dc (c, c_{opt}) \leq w_{opt}$.
\end{lemma}
\begin{proof}
Recall our assumption from Remark~\ref{rm:deltauniformity}.
By our assumption that at least one sample point lies on each facet,
MinBall cannot shrink past any facets of MaxBall($c_{opt}$).

Suppose for contradiction that $d_C(c,c_{opt})>w_{opt}$. Let $s$ be the point where a ray projected from $c$ through $c_{opt}$ intersects the boundary of MaxBall($c_{opt}$), and let $R$ denote the radius of the MinBall. Observe that $R$ must be large enough for \MinBall\ to contain $s$ and therefore $R\geq d_C(c,s)$.
\begin{align*}
R &\geq d_C(c, c_{opt}) + d_C(c_{opt}, s) \tag*{by collinearity} \\ 
  &> w_{opt} + d_C(c_{opt}, s) \tag*{by assumption} \\
  &= w_{opt} +  r_{opt} \tag*{by MaxBall($c_{opt}$).}
\end{align*}

Thus, since $w_{opt}+r_{opt}=R_{opt}$, we find $R>R_{opt}$, which is a contradiction since $R$ must be the smallest radius of the MinBall across all possible centers. 
Therefore, we have that $d_C(c,c_{opt})$ cannot be larger than $w_{opt}$. 
\end{proof}

Lemma~\ref{lm:searchspace} helps us constrain the region within which $c$ must be contained. Let us now reason about how a given center point, $c$, would serve as an approximation. For convenience, let us define $R:= |\MinBall(c)|$ and  $r:= |\MaxBall(c)|$ as the radii of the MinBall and MaxBall centered at $c$, respectively.

\begin{lemma} \label{lm:dist}
   Suppose $c$ is an arbitrary center-point in our search region,
   and the two directed distances between $c$ and $c_{opt}$ are at most~$t$,
   \textit{i.e.}, $t\geq \max \{ \dc (c,c_{opt}), \dc (c_{opt},c)\}$.
   Then, we have that $|\MWA(c)|\le w_{opt}+2t$.
\end{lemma}

\begin{proof}

Knowing that all sample points must be contained within the $\MWA$, the $\MWA(c)$ cannot expand past the furthest or closest point in $\MWA$ from $c$ under $d_C$. 
Let us now define these two points and use them to bound the radii for $\MinBall(c)$ and $\MaxBall(c)$.

Let $p$ be the point where the ray from $c$ through $c_{opt}$ intersects the boundary of MinBall($c_{opt}$). $\MinBall(c)$ cannot extend further than $p$. 
\begin{align*}
\dc(c,p) =\dc (c,c_{opt})+\dc (c_{opt},p) &\leq  t + \dc (c_{opt},p)\\
R &\leq R_{opt} +t.
\end{align*}
Conversely, let $q$ be the intersection point where the ray projected from $c_{opt}$ through $c$ intersects the boundary of MaxBall($c_{opt}$),
in which case $\MaxBall(c)$ cannot collapse further than $q$.
\begin{align*}
\dc (c,q) =\dc (c_{opt},q)-\dc (c_{opt},c) &\geq \dc (c_{opt},q) -t\\
r &\geq r_{opt}- t.
\end{align*}
Combining these bounds with the fact that $|MWA(c)|=R-r$ we find that
$|\text{MWA}(c)|\le w_{opt}+2t$. 

\end{proof}

For simplicity, let us consider two points $a,b$ to be \emph{$t$-close} 
(under $C$) whenever $t\geq \max \{ \dc (a,b), \dc (b,a)\}$.

\begin{lemma}
   \label{lm:capprox}
   If $c$ is the center of $\MinBall$, 
then $\MWA(c)$ is a constant factor approximation of the $\MWA$, that is,
$|\MWA(c)|\le b|\MWA|$, for some constant $b\ge1$, under translations.
\end{lemma}

\begin{proof}
From Lemma~\ref{lm:searchspace}, we have that $\dc (c, c_{opt}) \leq w_{opt}$. 
If $c$ and $c_{opt}$ are $w_{opt}$-close,
then we can directly apply the second part of Lemma~\ref{lm:dist} to find $r \geq r_{opt} - w_{opt}$ and $R\leq R_{opt}$,
such that $|\MWA(c)|\leq R_{opt}-(r_{opt} - w_{opt})$, thus proving that this is
a 2-approximation. 
If $\dc$ is a metric, then $\dc (c_{opt},c)=\dc (c, c_{opt})$ and this must always be the case.
However, if $\dc (c_{opt},c)> w_{opt}$, then we must use the Euclidean distance to find $\dc (c_{opt},c)$. Let vector $u:=c-c_{opt}$, and let us define unit vectors with respect to $\dc$ and $d_{\overline{C}}$, such that
\begin{gather*}
\widehat{u}_C=\frac{u}{\dc (c_{opt},c)},\hspace{.5cm} \widehat{u}_{\overline{C}}=\frac{\overline{u}}{\dc (c,c_{opt})}\\
||\widehat{u}_C|| \SP \dc (c_{opt},c) ~=~ ||u|| ~=~ ||\widehat{u}_{\overline{C}}|| \SP \dc (c,c_{opt})\\
\dc (c_{opt},c) \leq \frac{||\widehat{u}_{\overline{C}}||}{||\widehat{u}_C||}w_{opt} \tag*{from Lemma~\ref{lm:searchspace}.}
\end{gather*}
Under any convex distance function, 
$\frac{||\widehat{u}_{\overline{C}}||}{||\widehat{u}_C||}$ is bounded from above by 
$A=\max_{v\in\mathbb{R}^d} \frac{||\widehat{v}_{\overline{C}}||}{||\widehat{v}_C||}$,
which corresponds to finding the direction, $v$, of the largest asymmetry in $U_C$. 
Thus, by Lemma~\ref{lm:dist}, $|\MWA(c)|\leq (A+1) w_{opt}$. 
Under our (fixed) polyhedral distance function, $A$ is constant; hence,
MWA($c$) is a constant-factor approximation.
\end{proof}

\paragraph{\boldmath{$(1+\eps)$}-approximation.}\label{sec:approx-algorithm}
Let us now describe how
to compute a \epsapprox\ of $\MWA$.

We begin with Lemma~\ref{lm:epsapprox}, which defines how close to $c_{opt}$ is sufficient for a \epsapprox. In \Cref{thm:epsapprox}, we define a grid of candidate center-points so that any point in the search region has a gridpoint sufficiently close to it.

\begin{lemma}
   \label{lm:epsapprox}

Suppose $c_{opt}$ and $c$ 
are $({\eps w}/(2b))$-close,
where $w = |\MWA(c_M)|$, $c_M$ is the center of $\MinBall$,
and $b$ is the constant from Lemma~\ref{lm:capprox}.
Then, $\MWA(c)$ is a \epsapprox\ of $\MWA$ under translations.

\end{lemma}

\begin{proof}

It suffices to show that the width of our approximation only exceeds the optimal width by a factor of at most $(1+\eps)$.

Assuming $c$ and $c_{opt}$ are $t$-close, and using Lemma~\ref{lm:dist}, 
 we require that $w_{opt}+2t \leq (1+\eps)w_{opt}$,
   \textit{i.e.}, $t \leq \eps w_{opt}/2$.

   Let us then choose $t \leq {\eps w}/(2b)$, 
knowing that $w \leq b w_{opt}$ from Lemma~\ref{lm:capprox},
which is sufficient for achieving a $(1+\eps)$-approximation.

\end{proof}

Knowing how close our approximation's center must be, we can now present a \epsapprox\ algorithm to find a center satisfying this condition.
\begin{theorem}
\label{thm:epsapprox}
   One can achieve a $(1+\eps)$-approximation of the $\MWA$ under
translations in $O(\eps^{-d} n)$ time.
\end{theorem}

\begin{proof}
   The MinBall can be computed in $O(n)$ time~\cite{das2016}. 
   By Lemma~\ref{lm:searchspace}, we have that
   $d_C(c,c_{opt}) \leq w_{opt}$, where $c$ is the MinBall center. 
   This implies that $c_{opt}$ must lie within the placement $c+w_{opt} C$ or more generously in $P$, defined as $c+w C$.
   Furthermore, from Lemma~\ref{lm:epsapprox}, we know 
that being $({\eps w}/(2b))$-close to~$c_{opt}$ 
suffices for an $(1+\eps)$-approximation. 
Therefore, overlaying a grid $G$ that covers $P$, such that any point in $p\in P$ is $({\eps w}/(2b))$-close to a gridpoint, guarantees the existence of a point $g\in G$ for which $\MWA(g)$ is a $(1+\eps)$-approximation. 

   Since $P$ and $({\eps w}/(2b))$-closeness 
are both defined under $\dc$, we translate this to a cubic grid for simplicity.
Let $Q$ be the smallest cube enclosing $P$ and $q$ be the largest cube 
enclosed by $({\eps w}/(2b)) C$.
   Let us now define a grid, $G$, to span over $Q$ with cells the size of $q$.

   This grid, $G$, has $O(Fb/\eps)$ gridpoints per direction and $O(F^db^d\eps^{-d})$ gridpoints in total,
   where $F$ corresponds to the fatness of $C$ under the cubic distance function.

   Let us define the cubic distance function, $d_q$, with unit cube $U_q=q \cdot (2b)/(\eps w)$, such that $U_q$ is the largest cube enclosed by $C$. 
   
   The grid $G$ guarantees that for every point $p$, there exists a gridpoint $g\in G$ such that $d_q(p,g)\leq {\eps w}/(2b)$.
   Since the unit cube is contained within the unit polyhedron, we have that $\dc(a,b) \leq d_q(a,b)\ \forall a,b$; and since $d_q$ defines a metric, $p$ must also 
be $({\eps w}/(2b))$-close under~$\dc$.
   Finding the gridpoint providing the \epsapprox\ takes 
$O(F^db^d\eps^{-d}n)$ time,\footnote{For metrics, MinBall provides a 2-approximation, thus $b{=}2$. For non-metrics, we can remove this constant by first using this algorithm with $\eps{=}1$ in order to find a 2-approximation in linear-time, and using this approximation for gridding in the main step.}
 which, under a fixed $\dc$, is $O(\eps^{-d}n)$ time.
\end{proof}

%% file: faster.tex
\paragraph{Faster grid-search in two dimensions.}\label{sec:improving}
The algorithm of Theorem~\ref{thm:epsapprox} recalculates the MWA at every gridpoint. However, small movements along the grid should not affect the MWA much. We use this insight to speed up MWA recalculations for two dimensions. 

Let us first define the \emph{contributing edge} of a sample point, $p\in S$, as the edge of $C+g$ intersected by the ray emanating from 
a gridpoint, $g$, towards $p$. Under this center-point, $p$ will only directly affect the placement of the contributing edge. 
Observe that given vectors $\ora{v}\in C$, defined as the vectors directed from the center towards each vertex, the planar subdivision, created by rays for each $\ora{v}$ originating from $g$, separates points by their contributing edge.
For any two gridpoints, $g_1$ and $g_2$, and rays projected from them parallel to $\ora{v}$, any points within these two rays will contribute to different edges  under $g_1$ and $g_2$.
We denote this region as the \emph{vertex slab} of vertex $v$, 
and the regions outside of this as \emph{edge slabs}. 
Points within an edge slab contribute to 
the same edge under both gridpoints,
maintaining the constraints this imposes on the MWA, can therefore be achieved with the two extreme points per edge slab. 
If we consider vertex slabs for all $g\in G$, we must be able to quickly calculate the strictest constraints imposed by points in a subset of vertex slabs.
An example of the planar subdivision for two points is shown in 
Figure~\ref{fig:slabs}. 
\begin{figure}[tb]
\centering
\includegraphics[width=0.42\linewidth]{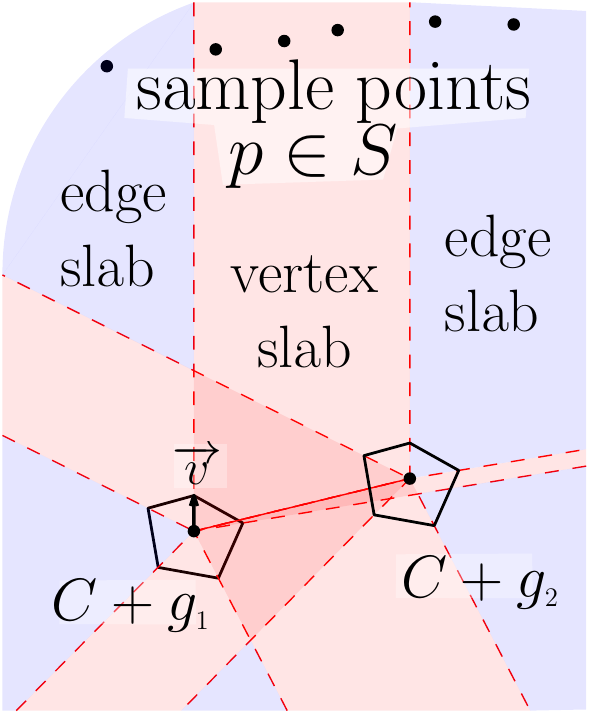}	
\caption{
Planar subdivision defining vertex slabs (red) and edge slabs (blue) for two candidate center-points, and showing membership of some sample points.
}
\label{fig:slabs}
\end{figure}

Given a grid $G$, we write $g_{i,j}\in G$ to be the gridpoint at index $(i, j)$. Consider the set of all grid lines $L_{v}$ defined by rays parallel to $\ora{v}$ starting at each gridpoint. $L_v$ defines a planar subdivision corresponding to the edge slabs between gridpoints.  Before attempting to identify the extreme points for each edge slab, we first need to find a quick way to identify the slab in $L_v$ that contains a given sample-point, $p$. 

\begin{lemma}
\label{lm:place}
For a specific vector $\ora{v}$ and an $m\times m$ grid, we can identify which slab contains a sample point, $p$, in $O(\log m)$ time with $O(m^2)$-time preprocessing.
\end{lemma}

\begin{proof}
Consider the orthogonal projection of grid lines in $L_v$ onto a line $\ora{v_{\bot}}$ perpendicular to $\ora{v}$, the order in which these lines appear in $\ora{v_{\bot}}$ defines the possible slabs that could contain $p$ (see \Cref{fig:pointloc}). We can project a given grid line $l\in L_{v}$ onto $\ora{v_{\bot}}$ in constant time. With the grid lines in sorted order, we can perform a binary search through the $m^2$ points in $O(\log m)$ time to identify the slab containing $p$.

Using general sorting algorithms, we could sort the grid lines in $O(m^2\log m)$ time. However, since these lines belong to a grid, we can exploit the uniformity to sort them in only $O(m^2)$ time. Consider the two basis vectors defining gridpoint positions $\hat{\imath}=g_{(1,0)}-g_{(0,0)}$ and $\hat{\jmath}=g_{(0,1)}-g_{(0,0)}$, and their sizes after orthogonal projection onto $\ora{v_{\bot}}$, $|\hat{\imath}_{\bot}|$, and $|\hat{\jmath}_{\bot}|$. Without loss of generality, assume that $|\hat{\imath}_{\bot}|\geq|\hat{\jmath}_{\bot}|$, in which case grid lines originating from adjacent gridpoints in the same row must be exactly $|\hat{\imath}_{\bot}|$ away. In addition, any region $|\hat{\imath}_{\bot}|$-wide, that does not start at a grid line, must contain at most a single point from each row. 
Furthermore, since points in the same row are always $|\hat{\imath}_{\bot}|$ away, they must appear in the same order in each region.

We can therefore initially split $\ora{v_{\bot}}$ into regions $|\hat{\imath}_{\bot}|$ wide. Sorting the grid lines $l\in L_v$ into their region can therefore be calculated in $O(m^2)$ time. Now we can sort the $m$ points in the region containing points from every row in $O(m\log m)$ time. Since each region has the same order, we can place points in other regions by following the order found in our sorted region, thus 
taking $O(m^2)$ preprocessing time for sorting the points.
\end{proof}

\begin{figure}
\begin{subfigure}{0.49\linewidth}
\centering
\includegraphics[width=\linewidth]{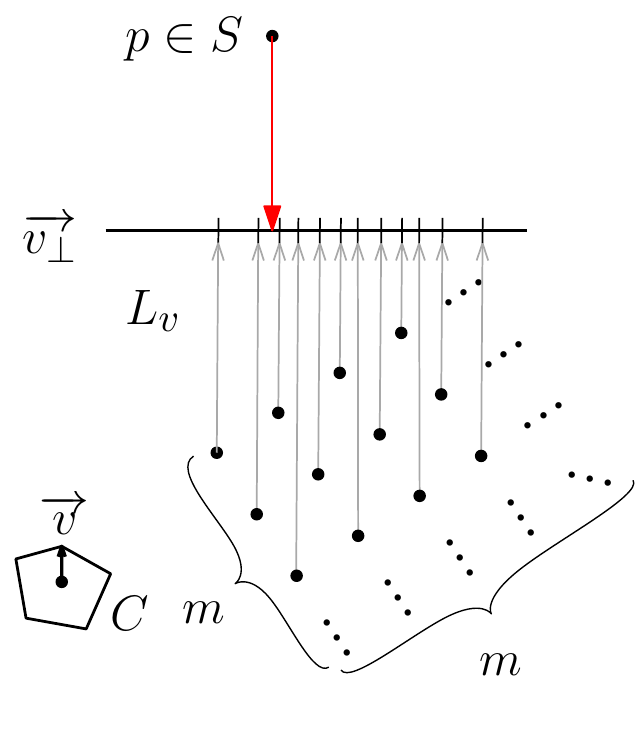}
\caption{A demonstration of the point location problem with the subdivision, $L_v$, and a visualization of the gridpoints and sample point projections onto $\ora{v_{\bot}}$.}
\label{fig:pointloc}
\end{subfigure}
\hfill
\begin{subfigure}{0.48\linewidth}
\centering
\includegraphics[width=\linewidth]{img/extrememost}
\caption{Finding the extreme points (red) under $\ora{e}_L$ in subdivision $L_v$ for each region (solid) and for all regions to its left (dashed).}
\label{fig:extrememost}
\end{subfigure}
\caption{A visual representation of the projections involved while point locating within the vertex slabs and while finding the extreme points in each slab.}
\end{figure}
Recall that points to the left of a given line $l\in L_v$ contribute to the edge to the left of $v$, \textit{i.e.}, all points belonging to slabs to the left of $l$. We can therefore isolate the points in these slabs causing the largest potential change in MWA.

\begin{lemma}
\label{lm:2dconst}
For a vertex $v \in C$ and grid line $l\in L_v$ through gridpoint $g$, let $l_L$ and $l_R$ refer to the slabs on the subdivision imposed by $L_v$ immediately to the left and right of $l$, respectively. Assuming $l_L$ maintains the points to the left of $l$ imposing the strictest constraints on $\MWA(g)$, and $l_R$ to the right, one can calculate $\MWA(g)$ in $O(1)$ time.
\end{lemma}

\begin{proof}
Finding $\min_{p\in S}  d_C(g,p)$ and $\max_{p\in S}  d_C(g,p)$ can now be achieved by optimizing only over the set of points in $\{ l_L \cup l_R, \, \forall v {\in} C\}$ and all points in edge slabs. This set would contain two points per vertex and two points per edge, yielding a constant number of points. Thus, MWA($g$) can be found in constant time.
\end{proof}
\begin{theorem}
\label{thm:2dapprox}
A \epsapprox\ of the $\MWA$ in two dimensions 
can be found in $O(n\log \eps^{-1} + \eps^{-2})$ time under translations.
\end{theorem}
\begin{proof}
For each vertex, $v$, we use Lemma~\ref{lm:place} to identify the slab for every sample point. 
For each slab, we maintain only the two extreme points for each of the edges incident on $\ora{v}$. Let $\ora{e}_L\in C$ denote the vector describing the edge incident on $\ora{v}$ from the left, and vice versa for $\ora{e}_R\in C$ incident from the right. 
For each slab, we maintain only points which when projected in the relevant direction, $\ora{e}$, cause the furthest and closest intersections with the boundary (shown for $\ora{e}_L$ in Figure~\ref{fig:extrememost}).

With a left-to-right pass, we update a slab's extreme points relative to $\ora{e}_L$ to maintain the extreme points for itself and slabs to its left. With a right-to-left pass, we do the same for $\ora{e}_R$ and maintain points in its slab and slabs to its right.

Thus, for each vertex, we create the slabs in $O(\eps^{-2})$ time, place every sample point in its slab in $O(n\log \eps^{-1})$ time, and maintain only the extreme points per slab in constant time per sample point. With $O(\eps^{-2})$ time to update each slab after processing the sample points, we can update the slabs such that they hold the extreme points across all slabs to their left or right (relative to $\ora{e}_L$ and $\ora{e}_R$, respectively).

For each edge slab, finding the extreme points is much simpler since finding $\min d_C(g,p)$ and $\max d_C(g,p)$ will always be based on the same contributing facet for all points within the same edge slab .

Thus, after finding the extreme points in both vertex slabs, we can calculate $\MWA(g)$ in constant time as described in Lemma~\ref{lm:2dconst}. Taking $O(\eps^{-2})$ time to find $\min_{g\in G} \MWA(g)$,  which by Theorem~\ref{thm:epsapprox} provides a \epsapprox\ of the minimum width annulus, and considering the $O(n\log \eps^{-1})$ pre-processing time completes the proof of the claimed time bound.
\end{proof}

%% file: rotations.tex
\section{Approximating MWA allowing rotations}
\label{sec:rotations}

In this section we consider rotations.
As with Lemma~\ref{lm:epsapprox}, our goal is to find the maximum tolerable rotation sufficient for a \epsapprox. 
Observe that when centered about the global optimum, the solution found under both rotation and translation is at least as good as the solution found solely through rotation (\textit{i.e.}, under a fixed center). We will therefore first prove necessary bounds for a \epsapprox\ under rotation only with the understanding that they remain when also allowing for translation.

Consider the polyhedral cone around $\ora{v}$, and define the \emph{bottleneck angle} as the narrowest angle between a point on the surface of the polyhedral cone and $\ora{v}$. Let~$\theta$ be the smallest bottleneck angle across all $\ora{v}\in C$. Let MWA$_{\alpha}(c$) denote the MWA centered at $c$, where C has been rotated by angle $\alpha$. Let us also use similar notations for MinBall and MaxBall. 

\begin{lemma}
\label{lm:scale}
Rotating by $\alpha$ causes $\MinBall_{\alpha}(c)$ to grow by at most $\frac{\sin(\pi-\theta-\alpha)}{\sin{\theta}}$ (and the reciprocal for $\MaxBall_{\alpha}(c)$).
\end{lemma}
\begin{proof}
Similarly to Lemma~\ref{lm:dist}, all sample points must be contained within $\MinBall(c)$. MinBall$_{\alpha}(c)$ can only expand to the furthest point within $\MinBall(c)$ under the new rotated distance function. 
Let us now consider the triangle formed between $c$, the vertex $v$ of the original  MinBall, $v_{0}$, and the rotated vertex $v_{\alpha}$ (shown in Figure~\ref{fig:scale}). Since our calculations focus towards the same vertex, we can work with Euclidean distances. The quantity $|v_0-c|$ defines the radius $r_1$ of the original polyhedron, and $r_2=|v_\alpha - c|$ the radius of the rotated one. With $\gamma=\pi-\theta-\alpha$ as the remaining angle in our triangle and using the sine rule, we find that
\begin{align*}
\frac{r_2}{r_1}=\frac{\sin \gamma}{\sin{\theta}}= \frac{\sin(\pi-\theta-\alpha)}{\sin{\theta}}.
\end{align*}
Observe that $\theta$ is the angle maximizing this scale difference. This applies to rotating by $\alpha$ in any direction about $\ora{v}$ (as shown in \Cref{fig:deflection}), and since this direction need not coincide with $\theta$, the scaled polyhedron might not touch the original.

For MaxBall$_{\alpha}(c)$ to be contained within MaxBall$(c)$, the same example holds after switching references to the scaled and original. In this case, $\theta$ minimizes $r_1/r_2$.
\begin{figure}[tb]
\begin{subfigure}{\linewidth}
\centering
\includegraphics[width=.75\linewidth]{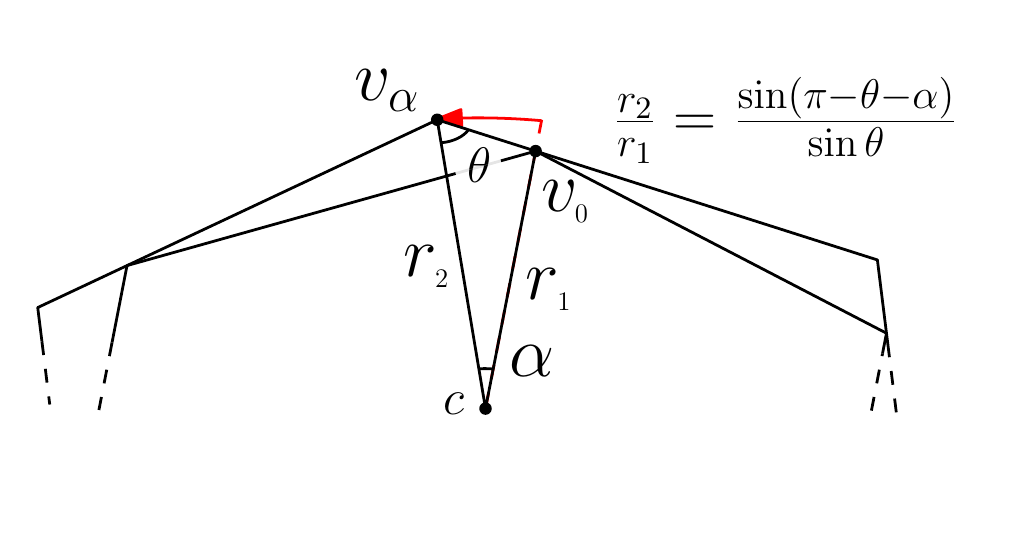}
\caption{A demonstration of the scale increase necessary for a polyhedron rotated by $\alpha$ to contain the original.}
\label{fig:scale}
\end{subfigure}
\begin{subfigure}{\linewidth}
\centering
\includegraphics[width=0.35\linewidth]{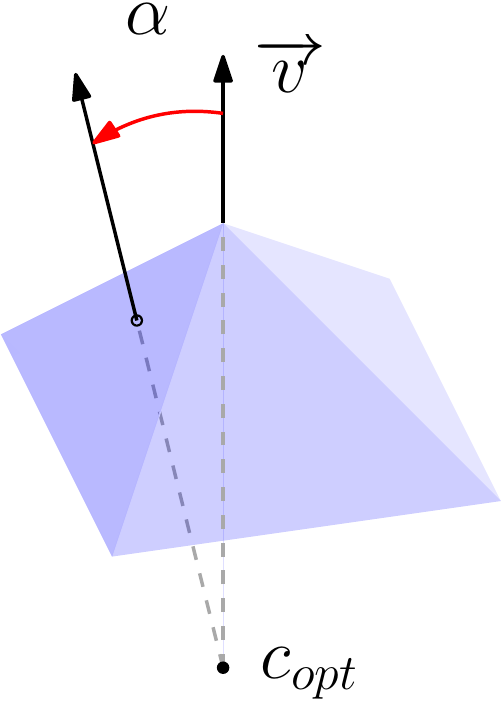}
\caption{A rotation by $\alpha$ in an arbitrary direction about $\ora{v}$.}
\label{fig:deflection}
\end{subfigure}
\caption{Visual representations for the effect of rotating by $\alpha$, demonstrating the scale increase and demonstrating how a rotation by $\alpha$ is defined for higher dimensions.}
\end{figure}
\end{proof}
Let us now determine the rotation from the optimal orientation that achieves a \epsapprox.
\pagebreak
\begin{lemma}
\label{lm:alpha}
Given a center $c$, we have that $\MWA_{\alpha}(c)$ is a \hbox{\epsapprox} when $\alpha$ is smaller than
\begin{align*}
  \arcsin\left( \frac{\sin\theta}{2f}\left(1{+}\eps \pm \sqrt{(1{+}\eps)^2+4f(f-1)}\right)\right)-\theta.
\end{align*}
\end{lemma}
\begin{proof}
Define $f$ as the ratio of the radius of MinBall($c_{opt}$) to $w_{opt}$
(\textit{i.e.}, $fw_{opt}=|\text{MinBall}(c_{opt})|$). Note that $f$ corresponds to the slimness of $S$ under $\dc$ over all rotations of $C$. Using Lemma~\ref{lm:scale}, we know that
\begin{gather}
|\text{MWA}_{\alpha}(c)|\leq \frac{\sin{\gamma}}{\sin{\theta}}|\MinBall(c_{opt})| - \frac{\sin{\theta}}{\sin{\gamma}}|\MaxBall(c_{opt})|\nonumber\\
\frac{\sin{\gamma}}{\sin{\theta}}f w_{opt} -\frac{\sin{\theta}}{\sin{\gamma}}(f{-}1)w_{opt}\leq (1{+}\eps) w_{opt}\label{eq:alpha1}\\
\frac{\sin{\gamma}}{\sin{\theta}}f-\frac{\sin{\theta}}{\sin{\gamma}}(f{-}1)\leq (1{+}\eps).\label{eq:alpha2}
\end{gather}
For a \epsapprox,  $|\MWA_{\alpha}(c)|\leq(1{+}\eps)w_{opt}$ imposing the right side of Relation~(\ref{eq:alpha1}), its left side follows by definition of $f$, and Relation~(\ref{eq:alpha2}) by cancellation of $w_{opt}$. Since $\theta$ is constant, we can rearrange the above into a quadratic equation and solve for $\sin \gamma$.
\begin{gather}
\sin{\gamma}=\frac{\sin\theta}{2f}\left(1{+}\eps \pm \sqrt{(1{+}\eps)^2+4f(f{-}1)}\right).\label{fm:in}
\end{gather}
However, $\arcsin$ will find $\gamma\leq \pi$, whereas we need the obtuse angle $\pi-\gamma$.

Thus, proving this lemma's titular bound, 

and achieving a \epsapprox.
\end{proof}

Let us now establish a more generous lower-bound that will prove helpful when developing algorithms.

\begin{lemma}
\label{lm:dir}
The angular deflection required for a \hbox{\epsapprox} is larger than 
${\theta\eps}/(2f)$.
\end{lemma}

\begin{proof}
Observe that $\gamma$ is of the form $\arcsin (k \sin \theta)$ and thus, in order for $\alpha=\gamma{-}\theta$ to be positive, we must have $\theta<\pi/2$ and $k>1$. We will prove this is the case.

\begin{gather}
k=\frac{1{+}\eps}{2f} + \sqrt{\left(\frac{1{+}\eps}{2f}\right)^2-\frac{1}{f}+1} \label{eq:dir1}\\
\sqrt{\frac{1}{4f^2}-\frac{1}{f}+1}= \left|1-\frac{1}{2f}\right|\label{eq:dir2}\\
k>\frac{1{+}\eps}{2f} + \left|1-\frac{1}{2f}\right|=1+\frac{\eps}{2f}.\hspace{.4cm}\label{eq:dir3}
\end{gather}
Equation~(\ref{eq:dir1}) follows from Equation~(\ref{fm:in}) after expanding. Relation~(\ref{eq:dir3}) follows after using Equation~(\ref{eq:dir2}) as a lower bound for the square root term in Equation~(\ref{eq:dir1}) since $\eps>0$ and $f>1$.
This allows us to bound $\arcsin\left(\left(1+\dfrac{\eps}{2f}\right)\sin\theta\right)$ by using a Taylor's series expansion to find $(1+k)\cdot\theta\leq\arcsin((1+k)\sin\theta)$, 

thus proving that the bound from Lemma~\ref{lm:alpha} is greater than $\frac{\theta\eps}{2f}$.
\end{proof}

\begin{lemma}
\label{lm:rotapprox}
For fixed rotation of $C$, assume we have an $O(g(n))$-time algorithm for the optimal minimum-width annulus under translation. 

We can find a\linebreak \epsapprox\ of the $\MWA$ under rotations and translations in $O(f^{d-1}\eps^{1-d}g(n))$ time. 
\end{lemma}

\begin{proof}
A $d$-dimensional shape has a $(d{-}1)$-dimensional axis of rotation. Let us evenly divide the unit circle into $k$ directions. Let us also define a collection of all possible direction combinations as a grid of directions. For each grid direction, rotate $C$ by the defined direction and calculate the MWA in $O(g(n))$ time. The optimal orientation must lie between the $(d{-}1)$-dimensional cube formed by $2^{d-1}$ grid directions. 

Therefore, as long as the diagonal is smaller than $\frac{\theta\eps}{f}$, there exists a grid direction within $\frac{\theta\eps}{2f}$ of the optimal orientation, which implies a \epsapprox\ by Lemma~\ref{lm:dir}. 
Thus, we can achieve a \epsapprox\ in time $O\left(g(n)\cdot \left(\frac{2\pi f\sqrt{d-1}}{\theta\eps}\right)^{d-1}\right)$, where $d$ and $\theta$ are constant under a fixed distance function $\dc$.

\end{proof}

With a fixed center, Lemma~\ref{lm:rotapprox} can be used to approximate $\MWA$ under rotations in $O(n f^{d-1}\eps^{1-d})$ time.

\begin{theorem}
One can find a \epsapprox\ of $\MWA$ under rotations and translations in $O(f^{d-1}\eps^{1-2d}n)$ time for $d{\geq}3$, and $O(f n \eps^{-1}\log \eps^{-1}+ f\eps^{-3})$ time for $d{=}2$.
\end{theorem}

\begin{proof}
Consider using an approximation algorithm (from Theorems~\ref{thm:epsapprox} or~\ref{thm:2dapprox}) instead of an exact algorithm as in Lemma~\ref{lm:rotapprox}. Let us define $(1+\xi)$ as the approximation ratio necessary from the subroutines in order to achieve an overall approximation ratio of $(1+\eps)$, such that $(1+\xi)^2= 1+\eps$. Since $\xi= \sqrt{1+\eps}-1$ and $0<\eps<1$, $\xi$ must be larger than $(\sqrt{2}-1)\cdot \eps$, and thus, we can always pick a value for $\xi$ which is $O(\eps)$ and achieves the desired approximation. Thus, by following Lemma~\ref{lm:rotapprox}, we can find a $(1+(\sqrt{2}-1)\cdot \eps)$-approximation using the $(1+(\sqrt{2}-1)\cdot \eps)$-approximation algorithm from Theorem~\ref{thm:epsapprox} to find a \epsapprox\ in $O(f^{d-1}\eps^{1-d}\cdot \eps^{-d}n)$ time. Alternatively, for two dimensions, we can instead use the algorithm from Theorem~\ref{thm:2dapprox} to find a \epsapprox\ in $O(fn \eps^{-1}\log \eps^{-1}+ f\eps^{-3})$ time.
\end{proof}